\documentclass[11pt,letterpaper,twoside,reqno,nosumlimits]{amsart}

\usepackage{etoolbox}

\patchcmd{\section}{\scshape}{\bfseries}{}{}
\makeatletter
\renewcommand{\@secnumfont}{\bfseries}
\makeatother

\usepackage[dvipsnames]{xcolor}
\patchcmd{\section}{\normalfont}{\normalfont\color{MidnightBlue}}{}{}
\patchcmd{\subsection}{\normalfont}{\normalfont\color{MidnightBlue}}{}{}

\makeatletter
\def\subsubsection{\@startsection{subsubsection}{3}%
\z@{.5\linespacing\@plus.7\linespacing}{-.5em}%
{\normalfont\bfseries}}
\makeatother

\usepackage{fancyhdr}
\usepackage{amsmath,amsfonts,amsbsy,amsgen,amscd,mathrsfs,amssymb,amsthm,mathtools,tensor}

\usepackage{amsthm}
\usepackage{amsmath}
\usepackage{amssymb}
\usepackage{dsfont}

\usepackage{amsfonts}
\usepackage{graphicx}
\usepackage{subfig}
\usepackage{color}
\usepackage{url}
\usepackage{epstopdf}
\usepackage{pdfsync}
\usepackage[colorlinks]{hyperref}
\usepackage{comment}
\usepackage{mathabx}
\usepackage{upgreek}

\hypersetup{
    linkcolor=blue,
}

\newlength{\fixboxwidth}
\usepackage{epsfig,amsbsy,graphicx,multirow}

\usepackage{ algorithm, algorithmic}

\renewcommand{\algorithmiccomment}[1]{\bgroup\hfill//~#1\egroup}

\usepackage[all]{xy}
\usepackage{accents}

 \numberwithin{equation}{section}

\def\L{\mathcal{L}}

\def\<{\big\langle}
\def\>{\big\rangle}

\def\argmin{{\operatorname{argmin}}}

\definecolor{red}{rgb}{0.9, 0, 0}

\newcommand\blfootnote[1]{%
  \begingroup
  \renewcommand\thefootnote{}\footnote{#1}%
  \addtocounter{footnote}{-1}%
  \endgroup
}

\usepackage[colorinlistoftodos]{todonotes}

\newtheorem{theorem}{Theorem}
\newtheorem{proposition}{Proposition}
\newtheorem{corollary}{Corollary}
\newtheorem{lemma}{Lemma}
\newtheorem{problem}{Problem}

\theoremstyle{definition}

\newtheorem{definition}{Definition}

\newtheorem{example}{Example}

\theoremstyle{remark}

\newtheorem{remark}{Remark}
\newtheorem{discussion}{Discussion}

\usepackage[
    backend=biber,natbib,
    style=authoryear,
    citetracker=true,
    maxcitenames=1
]{biblatex}

\AtEveryCitekey{\ifciteseen{}{\defcounter{maxnames}{4}}}

\begin{document}

\title[Aggregation of Pareto optimal models]{Aggregation of Pareto optimal models}
\author{Hamed Hamze Bajgiran, Houman Owhadi} 

\date{\today}
\blfootnote{Affiliation: Division of Computing and Mathematical Sciences (CMS), California Institute of Technology. Email: \href{mailto:hhamzeyi@caltech.edu}{hhamzeyi@caltech.edu} and \href{mailto:owshadi@caltech.edu}{owhadi@caltech.edu} }

\maketitle

\begin{abstract}
Pareto efficiency is a concept commonly used in economics, statistics, and engineering. In the setting of statistical decision theory, a model is said to be Pareto efficient/optimal (or admissible) if no other model carries less risk for at least one state of nature while presenting no more risk for others. 
How can you rationally aggregate/combine a finite set of Pareto optimal models while preserving Pareto efficiency?
This question is nontrivial because weighted model averaging does not, in general, preserve Pareto efficiency. This paper presents an answer in four logical steps: (1) A rational aggregation rule should preserve Pareto efficiency (2) Due to the complete class theorem, Pareto optimal models must be Bayesian, i.e., they minimize a risk where the true state of nature is averaged with respect to some prior. Therefore each Pareto optimal model can be associated with a prior, and Pareto efficiency can be maintained by aggregating Pareto optimal models through their priors. (3)  A prior can be interpreted as a preference ranking over models: prior $\pi$ prefers model A over model B if the average risk of A is lower than the average risk of B (where the average is taken with respect to the prior $\pi$). (4) 
A rational/consistent aggregation rule should preserve this preference ranking: If both priors $\pi$ and $\pi'$ prefer model A over model B, then the prior obtained by aggregating $\pi$ and $\pi'$ must also prefer A over B. 
Under these four logical steps, we show that all rational/consistent aggregation rules are as follows: Give each individual Pareto optimal model a weight, introduce a weak order/ranking over the set of Pareto optimal models,  aggregate a finite set of models S as
the model associated with the prior obtained as the weighted average of the priors of the highest-ranked models in S. This result shows that all rational/consistent aggregation rules must follow a generalization of hierarchical Bayesian modeling.
 Following our main result, we present applications to Kernel smoothing, time-depreciating models, social Choice theory, and voting mechanisms.
\end{abstract}

\section{Introduction}

The purpose of this paper is to characterize rational/consistent aggregation rules for Pareto efficient/optimal (admissible) models, i.e., answer the following question:   how can a decision-maker consistently aggregate the opinions of different experts or different Pareto efficient models into one single/aggregate Pareto efficient model?


For example, the decision-maker can be a financial planner who can use different models or expert opinions to create a portfolio of assets to maximize the expected profit of her portfolio. Given access to a set of  different experts or Pareto efficient models,
she must design a plan/rule on how to aggregate the different models/opinions to form a single final Pareto efficient model/opinion. 

More generally, employing Wald's Decision theoretic setting, the decision-maker may be interested in estimating some quantity of interest depending on some unknown parameter given data sampled from distribution depending on that parameter. In the process of selecting a decision rule (a plan/rule on how to use the observed data to estimate the quantity of the interest), the decision-maker observes the opinions or characteristics of some experts  (which, from now on, we simply refer to as experts). The decision maker's goal is then to aggregate all those experts into a single Pareto efficient model to form the final decision rule.

Our goal is to show, for all these examples,
that the aggregation plan/rule and the final Pareto efficient model have a simple form under the following consistency/rationality requirements (derived from   \cite{hamzeowhadi1}).

\begin{enumerate}
    \item Regardless of the set of observed experts, the decision-maker plays optimally. That is, she never plays a rule that, regardless of the true underlying parameter, leads to a higher loss than another rule.
    \item As a consequence of the complete class theorem, the decision-maker should find a minimizer of the loss function with respect to a single prior.
    \item By enabling comparisons between decisions rules/models through their average loss, a prior 
    can be interpreted as a preference ranking over the set of all decision rules/models. In this interpretation, the decision-maker should find the highest-ranked decision rule (which carries the lowest average risk). 
    \item If the decision-maker interprets a prior as a ranking over decision rules in the risk set, she should consistently aggregate the observed  experts. The form of consistency we use is the one introduced in \cite{hamzeowhadi1} and also has been mentioned with different names and purposes in the literature on case-based decision theory and social choice theory. That is, if the decision-maker observes a set of experts $A$ and another disjoint set of experts $B$ and form their respective priors $f(A)$ and $f(B)$. Then, the aggregated ranking induced by $f(A\cup B)$ over the set of decision rules (risk set) should preserve the ranking of $f(A)$ and $f(B)$. That is, for every two decision rules $d_1$ and $d_2$, if both $f(A)$ and $f(B)$ prefer $d_1$ over $d_2$, then  $f(A\cup B)$ should also prefer $d_1$ over $d_2$. 
    
\end{enumerate}

Our main result is to show that 
the consistency/rationality requirements described above can only be satisfied by the following simple  aggregation rule (which can be interpreted as a generalization of \emph{Hierarchical Bayes}).

\begin{enumerate}

    \item Select a weight function and a weak ordering over  experts.
    \item Identify the prior associated with each expert.
    \item For every subset of observed experts, find the average prior (by averaging the priors of highest ranked individual experts in the subset with respect to the weight).
    \item Finally, select a minimizer of the loss function with respect to the obtained average prior of step 3.

\end{enumerate}

Note that the weight used to form the average prior is independent of the subset of observed experts.

The organization of the paper is as follows. In section \ref{sec_main}, we present the decision-theoretic setting used to formalize our results. Then, we articulate the four main logical steps leading to our main results. Section \ref{sec_examles}  provides examples of the representation of the set of experts with applications and connections to the literature in statistics and social choice theory.

\section{Main Model and Results}\label{sec_main}

Let $\mathcal{E}$ be a set of experts. Depending on the application, we may assume that $\mathcal{E}$ to be a subset of a linear vector space. In that case, each expert $e \in \mathcal{E}$ may have different characteristics encoded in the coordinates of the vector $e$.
The goal of the decision-maker is to identify a (modeling) rule for aggregating experts by mapping the set of finite subsets of the set $\mathcal{E}$, which we denote by $\mathcal{E}^*$, to a set of models or decision rules $\mathcal{M}$.

\begin{definition}\label{def_AR}
Let $\mathcal{E}$ be a set of experts and $
\mathcal{M}$ be a set of models.  A \textbf{\emph{modeling rule}} on $\mathcal{E}$ is a function $f:\mathcal{E}^* \to \mathcal{M}$, that maps any finite subset of experts $A\in \mathcal{E}^*$ to a model $f(A)\in \mathcal{M}$.
\end{definition}

We will now use Wald's decision-theoretic setting to describe  $\mathcal{M}$.

\subsection{Identification of \texorpdfstring{$\mathcal{M}$}{Lg} in Wald's decision theoretic setting}

Let $(\mathcal{X}, \Sigma)$ be a measurable outcome space and $(\Theta, \Sigma_\Theta)$ be a measurable space of the possible states of nature, with $\mathcal{P}(\Theta)$ being the set of probability distributions on $(\Theta, \Sigma_\Theta)$.  Moreover, there is a class of probability measures $\{P_\theta: \theta\in \Theta\}$ such that whenever the true state is $\theta\in \Theta$, the distribution of observations $X\in\mathcal{X}$ is according to $P_\theta$. In other words, $(\mathcal{X}, \Sigma, P_\theta)$ is a probability space for every $\theta \in \Theta$.


 Wald's decision-theoretic setting is concerned with
 the problem of estimating some \textbf{\emph{quantity of the interest}} $q(\theta)$ in a space $Q$ given the observation of samples from the outcome space whose distribution depends on the true state of nature $\theta\in \Theta$ ($q:\Theta\to Q$). 
Let  $l:Q\bigtimes Q  \to\mathbb{R}$ be a \textbf{\emph{loss function}}  such that  $ l(x,y)\geq 0$ for all $x,y\in Q$, and  $ l(x,y)= 0$ whenever $x=y$.

\begin{definition} A \textbf{\emph{randomized model}} or a \textbf{\emph{randomized decision rule}} is a function $d:\mathcal{X}\times [0,1]\to Q$ such that $ l(q(\theta), d)$ is a measurable function on the measurable space $\large(\mathcal{X}\times [0,1],\sigma( \Sigma\times \mathcal{B}[0,1]
)\large)$ for every $\theta\in \Theta$, where $\mathcal{B}$ represents the Borel $\sigma$-algebra.  We denote the set of all randomized decision rules by $\Delta(D)$.

\end{definition}

For every randomized decision rule $d$, the decision-maker first $u$ according to the uniform distribution on $[0,1]$ and then estimates $q(\theta)$ according to the non-randomized decision rule $d(\cdot,u):\mathcal{X}\to Q$.

\begin{definition}
The \textbf{\emph{risk function}} $R_q: \Theta\bigtimes \Delta(D)\to \mathbb{R}$ is defined as the expected loss given the state $\theta$ and the decision rule $d$:

\begin{equation}
    R_q(\theta,d)=E_{X\sim P_\theta, u\sim U[0,1]}[l(q(\theta),d(X,u))].
\end{equation}

\end{definition}

From now on, we set $\mathcal{M}=\Delta(D)$: we  assume that the decision-maker  selects a randomized decision rule, i.e., we consider the modeling rule $f:\mathcal{E}^* \to \Delta(D)$.

We will now investigate four logical steps (rationality conditions) in the process of identifying a rule $f$. We will show that if all four steps are satisfied, then the modeling rule $f$  must be a simple weighted average.

\subsection{Step 1; Optimality/Admissibility}\label{step1}

The goal of a decision-maker is to select a decision rule minimizing some risk function. As in fig~\ref{fig_risk_shape}, regardless of the procedure employed to select a decision rule, a rational decision-maker should always select a rule that cannot be worst than another rule for all states of nature.  Otherwise, there is another estimator which provides less risk for at least one state of nature, and no more risk for others.

\begin{figure}[h]
    \centering
    \includegraphics[width=0.5\textwidth]{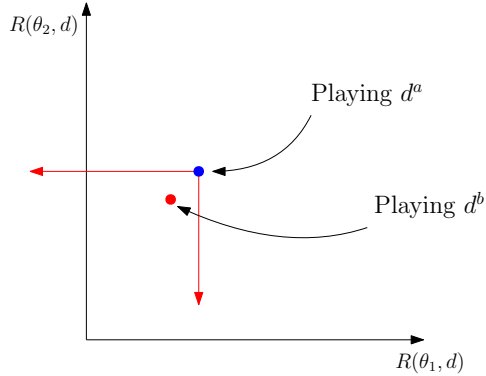}
    
    \caption{Playing $d^b$ is always better than playing $d^a$.}
    \label{fig_risk_shape}
\end{figure}

\begin{definition}
A decision rule $d_1\in \Delta(D)$ is as good as a decision rule $d_2\in \Delta(D)$ if $R_q(\theta,d_1)\leq R_q (\theta,d_2)$, $\theta\in \Theta$. A decision rule $d_2$ is \textbf{\emph{Pareto dominated}} by $d_1$ if $d_1$ is as good as $d_2$, and there exists at least a state $\theta$ such that $R_q(\theta,d_1)<R_q(\theta,d_2)$. 
An \textbf{\emph{admissible rule}} is a decision rule that is not Pareto dominated. A class of estimators $C\subset \Delta(D)$ is said to be \textbf{\emph{complete}} if it contains all admissible decision rules in $\Delta(D)$.

\end{definition}

The idea is that if the goal is to minimize the risk, the decision-maker should only use the models in a complete class. Therefore, we assume that the range of a good modeling rule is a subset of admissible (not Pareto dominated) decision rules.

\begin{definition} \label{assumption_1}
A modeling rule $f : \mathcal{E}^* \to \Delta(D)$ is \textbf{\emph{admissible}} if the range of $f$ is subset of the set of admissible (not Pareto dominated) randomized decision rules.
\end{definition}

\subsection{Step 2; Complete Class Theorem}\label{step2}
Admissible rules are related to the class of Bayes decision rules. To explore this relation, note that since the true state of nature $\theta$ is unknown, one may average the risk  with respect to a distribution of possible states of nature $\Theta$. The following definition captures this idea.

\begin{definition}
 The \textbf{\emph{Bayes risk function}} $R_q: \mathcal{P}(\Theta)\bigtimes \Delta(D)\to \mathbb{R}$ is the expectation of the risk function with respect to a \textbf{\emph{prior distribution}} $\pi\in \mathcal{P}(\Theta)$ and a randomized decision rule $d\in \Delta(D)$:
\begin{equation}
    R_q(\pi,d)=E_{\theta\sim\pi}[ R_q(\theta,d)],
\end{equation}
where, for ease of presentation, we have overloaded the notation of $R_q$.
\end{definition}

\begin{remark}\label{rem_linearity}
The Bayes risk function is a multi-linear function in the following sense. If the prior $\pi$ is a convex combination of two priors $\pi_1,\pi_2$, i.e $\pi=\alpha \pi_1+ (1-\alpha)\pi_2 $, then $R_q(\pi,d)=\alpha R_q(\pi_1,d)+(1-\alpha)R_q(\pi_2,d)$ for all $d\in \Delta(D)$. Moreover, if the distribution of a randomized decision rule $d$ is the same as the distribution of the randomization of two rules $d_1$ and $d_2$ which are selected with probability $\alpha$ and $1-\alpha$, then $R_q(\pi,d)=\alpha R_q(\pi,d_1)+(1-\alpha)R_q(\pi,d_2)$. 

One way of defining a randomized decision rule $d$ to have the distribution of the randomization of two randomized decision rules $d_1,d_2\in \Delta(D)$ with probability $\alpha$ and $1-\alpha$, is by defining

\begin{equation}
   d(x,u) =
    \begin{cases}
      d_1(x,\frac{u}{\alpha}) & \text{if } u< \alpha,\\
      d_2(x,\frac{u-\alpha}{1-\alpha}) & \text{otherwise}.
    \end{cases}
\end{equation}
\end{remark}

Bayes decision rules are the minimizer of the Bayes risk functions.

\begin{definition}\label{def:Sec_Bayes}
Let $\pi$ be a prior on $\Theta$. A \textbf{\emph{Bayes decision rule}} for the prior $\pi$ is a decision rule $d_\pi\in \Delta(D)$ that minimizes the Bayes risk function, i.e,  $d_\pi\in \operatorname{argmin}_{d\in \Delta(D)}R_q(\pi,d)$. 

\end{definition}

\cite{Wald:essentially} shows that in many cases, the class of Bayes decision rules forms a complete class. In other words, every admissible decision rule should minimize the loss function for a prior. Since, the result and the geometrical understanding of the result is important for the rest of the paper, we provide a geometrical overview and simplified proof of the result. To that end, it is helpful to consider the geometry of the \textbf{\emph{risk set}} $S\subseteq \mathbb{R}^\Theta$, which we endow with a topology later, defined as

\begin{equation}
    S=\{s\in \mathbb{R}^\Theta\ |\  \exists\  d\in \Delta(D)\text{ s.t } s(\theta)=R_q(\theta,d) \text{ for all } \theta\in \Theta \}.
\end{equation}

Essentially, for every \textbf{\emph{risk profile}} $s\in S$ there exists a randomized decision rule $d\in \Delta(D)$ such that the risk of playing the rule $d$ is exactly $s(\theta)$ for every state $\theta$. In other words, the risk set captures all possible attainable risk profiles.

By the definition of the risk set

\begin{equation}\label{eq:risk_eq}
\inf_{d\in \Delta(D)} R(\pi,d)=\inf_{s\in S} \int_{\Theta}s(\theta)\,d\pi(\theta).
\end{equation}
Informally, the complete class theorem is supported by a simple geometric argument. Minimizing the risk function defined by a prior $\pi$ over the set of randomized decision rules is the same as minimization of the linear function $\int_{\Theta}s(\theta)\,d\pi(\theta)$, defined by $\pi$, over the risk set $S$. By Remark~\ref{rem_linearity}, $S$ is a convex set. Therefore, the minimizer is on the intersection of the hyperplane defined by the prior $\pi$ and the boundary of $S$. As in fig~\ref{fig_optimal}, we will show that the other direction works as well. That is, we show that if the risk set is closed, then risk profiles associated with admissible decision rules are on the lower boundary of the risk set. Moreover, for every point on the lower boundary of the risk set $S$, there exists a tangent hyperplane defined by a prior $\pi$ to the risk set at that point. We show that the decision rule associated with that point on the boundary is a Bayes decision rule with respect to $\pi$.

\begin{figure}[ht]
    \centering
    \includegraphics[width=0.7\textwidth]{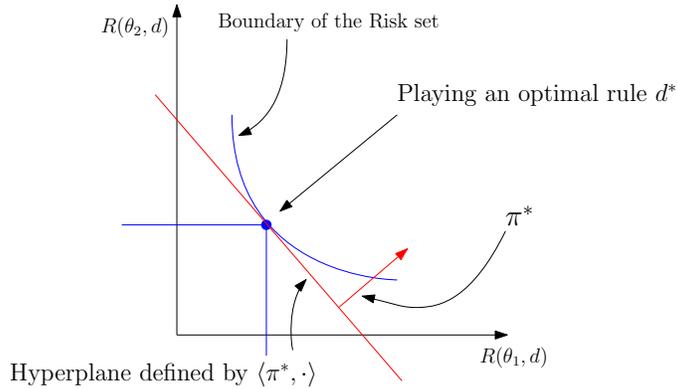}
    
    \caption{For every admissible decision rule $d^*$ on the boundary of the risk set, one can find a prior $\pi^*$ such that the hyperplane defined by $\langle \pi^*, \cdot\rangle $ that passes through $d^*$ separates the risk set and the set of negative functions.}
    \label{fig_optimal}
\end{figure}

To formalize this, we need the following definitions. Given a function $r\in \mathbb{R}^\Theta$, define the negative quadrant at $r$ to be 
\[Q_r=\{f\in \mathbb{R}^\Theta| f(\theta)\leq r(\theta) \text{ for all } \theta\in \Theta \}.\]
We define the lower boundary $\L(S)$ of $S$ by
$$\L(S)=\{r\in \mathbb{R}^\Theta| Q_r\cap \Bar{S}=\{r\}\} ,$$
where $\Bar{S}$ is the closure of the set $S$.
The set $S$ is said to be \emph{\textbf{closed from below }}if $\L(S)\subseteq S$. 

The main connection between a prior in the minimization of the risk function and a tangent hyperplane to the risk set is through the Riesz–Markov–Kakutani representation theorem (see \cite{Aliprantis2006} chapter 13). 

\begin{theorem}
Let X be a compact Hausdorff space and let $C(X)$ denote the set of continuous functions on $X$ equipped with $\sup$-norm. For any continuous linear function $\psi$ on $C(X)$, there is a unique signed Borel measure $\mu$ on $X$ such that
$$ \psi (f)=\int _{X}f(x)\,d\mu (x),     \qquad \forall f\in C(X).$$
The norm of $\psi$ as a linear function is the total variation of $\mu$, that is
$ \|\psi \|=|\mu |(X)$.
Finally, $\psi$ is positive \big($\psi(f)\geq 0$ for every non-negative function $f\in C(X)$\big) if and only if the measure $\mu$ is non-negative.

\end{theorem}

We are now ready to establish a complete class theorem; under some conditions, the set of Bayes decision rules contains the set of admissible rules. There are three main geometrical components, the convexity of the risk set, application of the separating hyperplane theorem to form a tangent hyperplane at any lower boundary of the risk set, and application of the Riesz–Markov–Kakutani representation theorem to obtain a representation of the tangent hyperplane in the form of an integral of the risk function with respect to a prior. 

Let $C(\Theta)$ be the space of continuous function on $\Theta$ equipped with the sup norm. In the following theorem, we assume that risk functions are continuous in their first argument and therefore $S\subset C(\Theta)$. Hence, we endow the risk set with the topology of $C(\Theta)$.

\begin{theorem} (Complete Class Theorem)\label{thm_main theorem_complete class }
Let $\Theta$ be a compact subset of a Hausdorff topological space.  If for every decision rule $d\in \Delta(D)$ the risk function $R_q(\theta,d)$ is a continuous function of $\theta$, and the risk set is closed from below in $C(\Theta)$, then the Bayes decision rules form an essentially complete class. 
\end{theorem}

\begin{proof}

Let $d\in \Delta(D)$ be an admissible rule and let $r(\cdot)=R(\cdot,d)\in S$ be its associated risk profile. By the admissibility of $d$, we have $(Q_r\cap C(\Theta))\cap \bar{S}=r$ and therefore any risk profile associated with an admissible rule is on the lower boundary $\mathcal{L}(S)$ of the risk set. 

Since $S$ and $Q_r\cap C(\Theta)$ are convex, $Q_r\cap C(\Theta)$ has a nonempty interior and $(Q_r\cap C(\Theta))\cap \bar{S}=r$, the separating hyperplane theorem (check \cite{Aliprantis2006} Section 5.13 or \cite{luenberger1969optimization} Thm 2. Section 5.12) implies that there is a continuous linear function $\psi$ separating $Q_r\cap C(\Theta)$ and $S$  achieving its minimum on the set $S$ at $r$. That is,

\begin{equation}\label{eq:sep}
     \sup_{f\in Q_r\cap C(\Theta)}\psi(f)\leq \psi(r)=\min_{s\in S}\psi(s).
\end{equation}

Since $\Theta$ is compact, the Riesz–Markov–Kakutani representation theorem assures us that there exists a finite signed measure $\mu_\psi$ on $\Theta$ representing the continuous linear function $\psi$ as

\begin{equation}\label{eq:rep}
     \psi (f)=\int_{\Theta}f(\theta)\,d\mu_\psi (\theta) \, ,  \qquad \forall f\in C(\Theta).
\end{equation}

To show that $\mu_\psi$ is a non negative measure, by the second part of the Riesz–Markov–Kakutani representation theorem, it is enough to show that  $\psi(g)\geq 0$, for every positive function $g\in C(\Theta)$. Assume that it is not the case and there exists a positive function $g\in C(\Theta)$ with $\psi(g)< 0$. Let $g_\alpha = -\alpha g+r$ for $\alpha >0$. By the positivity of $g$, $g_\alpha\in Q_r\cap C(\Theta)$ for every $\alpha >0$. Moreover, by the linearity of $\psi$, $\psi(g_\alpha) = \psi(-\alpha g + r)=-\alpha\psi(g) + \psi(r) > \psi(r)$ for every $\alpha > 0$.  However, by the choice of $\psi$ as in \eqref{eq:sep}, we should have $\psi(g_\alpha)\leq \psi(r)$, which is a contradiction. Therefore, the measure $\mu_\psi$ is a finite non negative measure, and by normalizing it we can assume, without loss of generality, that it is a probability measure.


Finally, observe that \eqref{eq:sep} and \eqref{eq:rep} imply that

$$r\in \operatorname{argmin}_{s\in S}\psi(s)=\operatorname{argmin}_{s\in S} \int_{\Theta}s(\theta)\,d\mu_\psi (\theta)$$
and \eqref{eq:risk_eq} implies that

$$\min_{\alpha\in \Delta(D)} \int_{\Theta}R(\theta,\alpha)\,d\mu_\psi (\theta)=\min_{s\in S} \int_{\Theta}s(\theta)\,d\mu_\psi (\theta).$$
Consequently, since $R(\cdot,d)=r$, we have

$$\min_{\alpha\in \Delta(D)} \int_{\Theta}R(\theta,\alpha)\,d\mu_\psi (\theta)=\min_{s\in S} \int_{\Theta}s(\theta)\,d\mu_\psi (\theta)= \int_{\Theta}r(\theta)\,d\mu_\psi (\theta)=R(\mu_\psi,d).$$
Hence, the decision rule $d$ is a Bayes decision rule with respect to the probability measure $\mu_\psi$ on $\Theta$. This completes the proof.
\end{proof}
In the more general case, such as where $\Theta$ is not compact, or the risk set is not closed from below, the Bayes decision rules do not necessarily form a complete class. However, similar geometrical arguments give us insight regarding the form of admissible rules. In many of the more general cases, admissible rules are limits of Bayes decision rules or are the minimizers of the risk function with respect to measures that are not necessarily finite measures. 

\begin{remark}\label{remark_technical_dominated masures}
Note that in many cases, such as when $P_\theta$ is an exponential family, and the loss function is a squared loss, the assumptions of the theorem are satisfied. More generally, if the loss function is continuous in its first argument, the quantity of the interest $q$ is continuous, $P_\theta$ are absolutely continuous with respect to the Lebesgue measure, and the density functions associated with $P_\theta$ are continuous in $\theta$ for every $x$, then the risk function is a continuous function of $\theta$ for every selected decision rule. 

\end{remark}

\begin{discussion}
Every admissible rule is the best response to a prior by the complete class theorem. However, it is not trivial to check whether a rule is admissible or not. To be precise, if the decision-maker has access to the set of rules, it is not trivial to check whether she is working with the admissible ones or not. 

For example, in the case that $x\sim \mathcal{N}_d(\theta, I)$ where $d$ is the dimension of the parameter space, one might consider the sample average as their rule. However, as the parameter space dimension becomes larger ($d\geq3$), the shrinkage-based estimator will beat the sample average for the mean square loss function. 

To be precise, in the case of single observation $x_1\sim \mathcal{N}_d(\theta,I)$, the James-Stein estimator
\begin{equation}
 \hat{\theta}_{JS}=\left(1-{\frac {(d-2)}{\|{\mathbf{x_1} }\|^{2}}}\right){\mathbf {x_1} },
\end{equation}
is going to dominate $x_1$ with respect to the Mean Square Loss function. The more interesting observation is that another class of rules can dominate the James-Stein estimator itself and as a result it is not admissible as well. 

However, we are assuming that the knowledge of the complete class theory makes us model as if we are minimizing a loss function with respect to a prior. The assumption might be incorrect in practice. 

Generally speaking, for a given decision rule $d$, we can not go over all the priors to check if it is admissible or not. However, one might minimize the loss function with respect to a prior, check the Bayes risk of that prior, and compare it with the risk given by the decision rule $d$. Accepting a rule as an approximately admissible is another question that is not our primary concern in this paper.
\end{discussion}

\subsection{Step 3; Interpreting a prior as a preference ranking over the risk set}\label{step3}

Lets elaborate more on the consequence of the complete class theorem. Again, the interpretation is through the following geometrical picture. Every prior $\pi\in \mathcal{P}(\Theta)$ induces a continuous linear functional $\langle\pi, s \rangle = \int_{\Theta}s(\theta)\,d\pi(\theta)$ over the risk set $\mathcal{S}=\{s\in \mathbb{R}^\Theta\ |\  \exists\  d\in\Delta(D)\text{ s.t } s(\theta)=R(\theta,d) \text{ for all } \theta\in \Theta \}$. The induced linear functional $\langle \pi,\cdot\rangle: \mathcal{S}\to \mathbb{R}$ is ranking the elements of the risk set with respect to their risk associated with the prior $\pi$. As a result of the complete class theorem, every admissible model is associated with the highest ranked point in the risk set with respect to a ranking associated with a prior $\pi \in \mathcal{P}(\Theta)$. 

Therefore, one may think about an admissible modeling rule  $f:\mathcal{E}^* \to \Delta(D)$ as a minimizer of the induced rankings of priors over the risk set $\mathcal{S}$. That is, for every $A \in \mathcal{E}^*$, there exists a $\pi_A\in \mathcal{P}(\Theta)$ that can be interpreted as a ranking over the risk set $\mathcal{S}$. The final model, $f(A)$,  is the one that has the highest-ranked over all the elements of $\mathcal{S}$ with respect to the ranking induced by $\pi_A$. Formally, we can define the connection as follows.

\begin{definition}
Let $f:\mathcal{E}^* \to \Delta(D)$ be an admissible modeling rule. A \textbf{\emph{ranking rule}} is a function $g_f: \mathcal{E}^* \to \mathcal{P}(\Theta)$ such that $f(A) \in \argmin_{d \in \Delta(D) } R(g_f(A), d)$ for every $A \in \mathcal{E}^*$. 
\end{definition}

To emphasis the view that each prior ranks the risk set linearly and for the simplicity of the notation, for every prior, we define a weak order $\succsim$ as follows:

\begin{definition}
Let $\pi\in \mathcal{P}(\Theta)$ be a prior over the set of states of nature. The weak order (reflexive, transitive, and complete binary relation) $\succsim_\pi$ over the risk set $\mathcal{S}=\{s\in \mathbb{R}^\Theta\ |\  \exists\  d\in\Delta(D)\text{ s.t } s(\theta)=R(\theta,d) \text{ for all } \theta\in \Theta \}$, is defined as:

\begin{equation}
s_1\succsim_\pi s_2 \Leftrightarrow \langle\pi,s_1\rangle\leq \langle\pi,s_2\rangle
\end{equation}

\end{definition}

As a consequence, for every $A\in \mathcal{E}^*$ we may interpret the prior $g_f(A) \in \mathcal{P}(\Theta)$ as its associated weak order $\succsim_{g_f(A)}$ over the elements of the risk set $\mathcal{S}$. Suppose we accept this viewpoint as a viable one. In that case, we may interpret the ranking rule $g_f$ as a ranking mechanism in which, by observing a subset of  experts $A\in \mathcal{E}$, ranks the risk set $\mathcal{S}$ using the induced weak order $\succsim_{g_f(A)}$. 
\begin{discussion}
In practice, we might not have a flat indifference curve. In other words, the assumption that the decision-maker may have a linear variety as her indifference set in the risk set might not be a viable one. There are approaches to handle these issues; however, it is not the primary concern in the paper. 

\end{discussion}

\begin{discussion}
Generally, there is no bijection between the set of priors and the set of admissible rules. If the boundary of the risk set is smooth enough (such that the set of sub-differentials have a unique element), then we can form a bijection between the set of decision rules and priors. Otherwise, as in fig ~\ref{fig_ambiguity}, an admissible rule may be the best response to distinct priors. In those cases, there will be an ambiguity between the selection of a prior that in a better way captures the main characteristics the modeler is interested in.  However, in this paper, we will not deal with such situations, and we will only assume that the modeler always selects a prior (this prior leads to a ranking of the risk profiles in the risk set). The process used to select the prior can be arbitrary.

\begin{figure}[ht]
    \centering
    \includegraphics[width=0.7\textwidth]{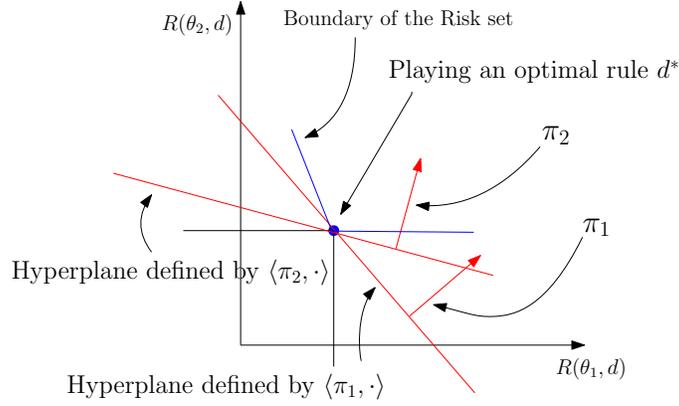}
    
    \caption{The decision rule $d^*$ is the best response with respect to both $\pi_1$ and $\pi_2$.}
    \label{fig_ambiguity}
\end{figure}

\end{discussion}

\subsection{Step 4; Consistency}\label{step4}

As a consequence of the last three steps, we reduce the problem from the class of modeling rule $f:\mathcal{E}^*\to \mathcal{M}$, to the class of ranking rules $g_f:\mathcal{E}^*\to \mathcal{P}(\Theta)$. In this step, our goal is to assess how to combine a result of two separate ranking orders $g_f(A)$ and $g_f(B)$ to form $g_f(A\cup B)$.

Lets consider the following simple example. Let $A=\{e_1,e_2\}$ with $g_f(e_1)=\pi_1$ and $g_f(e_2)=\pi_2$.  As a result of discussions in the step 3, to form $g_f(A)$ we may consider the aggregation of the corresponding weak orders $\succsim_{\pi_1}, \succsim_{\pi_2}$ over the risk set $\mathcal{S}$. One might assume that if both weak orders prefer a risk profile $s_1$ to another risk profile $s_2$, the aggregated ranking should also respect this order. If this is the case, we call the ranking to be \textbf{\emph{consistent}} with respect to both priors $\pi_1$ and $\pi_2$. For a general ranking mechanism, we may generalize the definition as follows.

\begin{definition}
Let $g_f:\mathcal{E}^*:\to\mathcal{P}(\Theta)$ be a ranking rule over the risk set $\mathcal{S}$. We say that $g_f$ is \textbf{\emph{weakly consistent}} if for every two disjoint sets $A, B\in \mathcal{E}^*$, and for every two risk profiles $s_1,s_2 \in \mathcal{S}$,

\begin{equation}\label{def_first_expr}
s_1 \succsim_{g_f(A)} s_2\ , \ s_1\succsim_{g_f(B)}  s_2 \Rightarrow s_1 \succsim_{g_f(A\cup B)}  s_2
\end{equation}

Moreover, it is \textbf{\emph{consistent}} if it also satisfies the following condition:
\begin{equation}\label{def_second_expr}
s_1 \succ_{g_f(A)} s_2\ , \  s_1 \succsim_{g_f(B)}  s_2 \Rightarrow s_1 \succ_{g_f(A\cup B)}  s_2
\end{equation}
\end{definition}

To better understand the geometry of the consistency, let $g_f:\mathcal{E}^*\to\mathcal{P}(\Theta)$  be a consistent ranking rule. Consider two disjoint subsets of experts $A, B\in \mathcal{E}^*$ and two risk profiles $s_1,s_2\in S$ such that $s_1 \succsim_{g_f(A)} s_2, \ s_1\succsim_{g_f(B)}  s_2$. Using the definition of the $\succsim_{g_f(A)}$ and $\succsim_{g_f(B)}$, we obtain that $\langle g_f(A),s_1-s_2\rangle\geq 0$ and $ \langle g_f(B),s_1-s_2\rangle\geq 0 $. Consistency implies that $s_1 \succsim_{g_f(A\cup B)} s_2$. Therefore, we should have  $\langle g_f(A\cup B),s_1-s_2\rangle\geq 0 $. Using the duality (Farkas' Lemma for finite dimensional cases or Hahn-Banach Theorem for general cases), the continuous linear function represented by $g_f(A\cup B)$ should be in the cone generated by $g_f(A),g_f(B)$ in the dual space of $C(\Theta)$. However, since $g_f(A\cup B)$ is a probability distribution, it should be a convex combination of $g_f(A)$ and $g_f(B)$. That is it is a randomization of $g_f(A)$ and $g_f(B)$ by some positive weight. Note that the condition~\ref{def_second_expr} in the definition of consistency, guaranteed that $f(A\cup B)$ should be in the interior of the line segment connecting $g_f(A)$ and $g_f(B)$ in the dual space of the risk set. Therefore, we may connect the consistency to another condition that has been studied in different litterateur with different names.

\begin{definition}We say that a ranking rule $g_f: \mathcal{E}^*\to \mathcal{P}(X)$ satisfies the \emph{\textbf{weighted averaging}} property if for all
 $A,B\in \mathcal{E}^*$ such that $A\cap B=\emptyset$, it holds true that
\begin{equation}\label{eqhejhdgeydg}
g_f(A\cup B)=\lambda g_f(A)+(1-\lambda)g_f(B)
\end{equation}
for some $\lambda \in [0,1]$ (which may depend on $A$ and $B$). We say that $f$ satisfies the \emph{\textbf{strict weighted averaging}} property if
\eqref{eqhejhdgeydg} holds true for $\lambda \in (0,1)$.
\end{definition}

Therefore, as a result of the duality, the two conditions are the same. 

\begin{lemma}\label{lem_consistencyeqwa}
Let $g_f:\mathcal{E}^*\to\mathcal{P}(\Theta)$ be a ranking rule. Then, the followings are equivalent:
\begin{enumerate}
\item $g_f$ is consistent.
\item $g_f$ satisfies the strict weighted averaging axiom.
\end{enumerate}
Moreover, the followings are also equivalent:
\begin{enumerate}

\item $g_f$ is weakly consistent.
\item $g_f$ satisfies the weighted averaging axiom.

\end{enumerate}
\end{lemma}

To elaborate more on the above observation, let $g_f$ be a consistent ranking rule and $A=\{e_1, \ldots, e_n\}\in \mathcal{E}^* $. By applying the result of the lemma \ref{lem_consistencyeqwa}, $g_f(A)$ is in the convex hull of the probability measures $g_f(e_i),\ i\in \{1, \ldots, n\}$. That means, there exists a randomization of $g_f(e_i)$ by a probability measure represented by $(\lambda_{e_1}, \ldots, \lambda_{e_n})\in \mathcal{P}(\{e_1,\ldots,e_n\})$ such that $g_f(A) = \sum_i \lambda_{e_i} g_f(e_i)$. Consequently lemma \ref{lem_consistencyeqwa} results in $g_f(A)\in \operatorname{ConvexHull}(g_f(e_i))$, with $e_i\in \mathcal{E}$, which we call a \emph{\textbf{coordinate wise Pareto}}.

\begin{definition}We say that a ranking rule $g_f: \mathcal{E}^*\to \mathcal{P}(X)$ is \emph{\textbf{coordinate wise Pareto}} if for all
 $A\in \mathcal{}{E}^*$, 
 
\begin{equation}
g_f(A)\in\operatorname{ConvexHull} \{g_f(e)| \ e\in A\}
\end{equation}

\end{definition}

We can have a better understanding of the Pareto through the lenses of duality.

\begin{lemma}
The rule $g_f: \mathcal{E}^*\to \mathcal{P}(X)$ is coordinate wise Pareto if and only if
for every set $A=\{e_1,\ldots,e_n\} \in {E}^*$, and for every two risk profiles $s_1,s_2 \in \mathcal{S}$,

\begin{equation}
s_1 \succsim_{g_f(e_1)} s_2\ , \ldots, \ s_1\succsim_{g_f(e_n)}  s_2 \Rightarrow s_1 \succsim_{g_f(A)}  s_2
\end{equation}
\end{lemma}
A simple induction shows that all consistent ranking rules are coordinate-wise Pareto.

\begin{corollary}
Every consistent ranking rule $g_f: \mathcal{E}^*\to \mathcal{P}(X)$ is a coordinate wise Pareto ranking rule.
\end{corollary}

One might wonder whether the opposite of the above observation is also true or not. In other words, whether the consistency is only about $g_f(A)$ being in the convex hull of the $g_f(e_i)$, for $e_i\in A$, or not. The answer is no. Consistency is a stronger assumption. To better understand it, consider the following example.

\begin{example}\label{Example_c_vs_cwp}

To elaborate more on the above observation, let $g_f$ be a coordinate wise Pareto rule and $A=\{x, y, z\}, B = \{x,y,w\}\in \mathcal{E}^* $ with $z \neq w$. Hence, there exists a two set of randomization $\lambda^A, \lambda^B$ on elements of $A, B$ such that, $g_f(A) = \lambda^A_{x} g_f(x) +\lambda^A_{y} g_f(y)+\lambda^A_{z} g_f(z)$ and $g_f(B) = \lambda^B_{x} g_f(x) +\lambda^B_{y} g_f(y)+\lambda^B_{w} g_f(w)$. 

Without the consistency, there is nothing more to be said. However, with consistency there is a connections between $ \lambda^A_{x}/  \lambda^A_{y}$  and $\lambda^B_{x}/ \lambda^B_{y}$. And the connection is that it is always possible to make $ \lambda^A_{x}/  \lambda^A_{y} = \lambda^B_{x}/ \lambda^B_{y}$.

To be more precise, consider the figure \ref{fig_example_1_1}. We will show that by knowing $g_f(x,y), g_f(z,y), g_f(z,w)$, we can deduce $g_f(x,y,z), g_f(x,y,w)$ uniquely.

\begin{figure}[ht]
    \centering
    \includegraphics[width=0.5\textwidth]{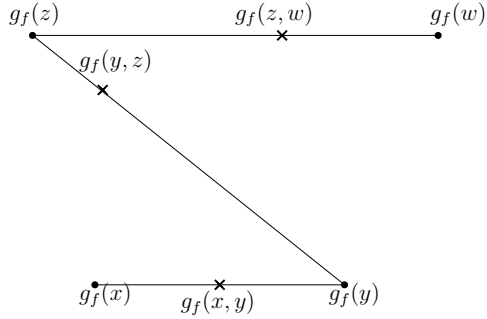}

    \caption{We are assuming that the value of $g_f$ is known at $\{x\}, \{y\}, \{z\},\{w\},\{x,y\}, \{y,z\},\{z,w\}$. The goal is to find $g_f(x,y,z),g_f(x,y,w)$ in a unique way.}
    \label{fig_example_1_1}
\end{figure}

\begin{figure}[htbp]
\subfloat[$g_f(x,y,z)$.]{\label{fig_example_1_2}\includegraphics[width=.45\linewidth]{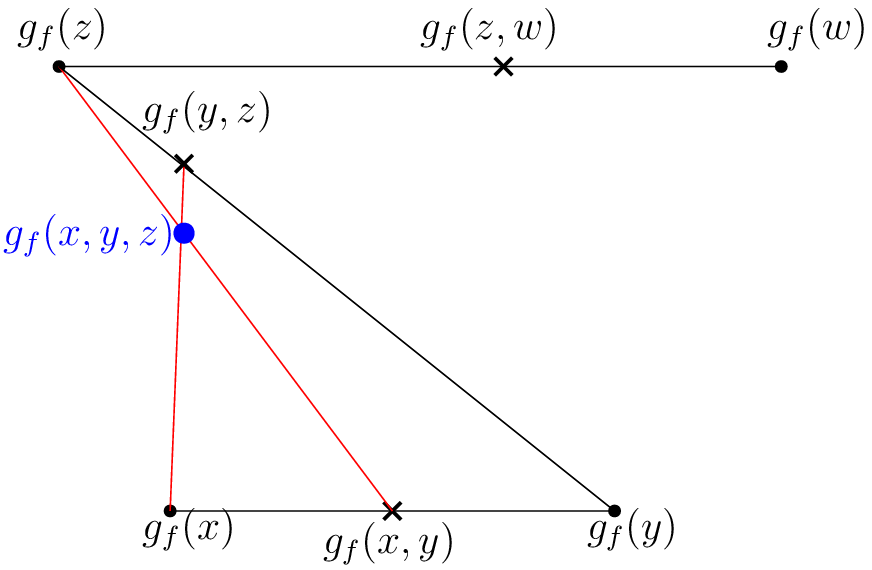}}\hfill
\subfloat[$g_f(x,y,z,w)$.]{\label{fig_example_1_3}\includegraphics[width=.45\linewidth]{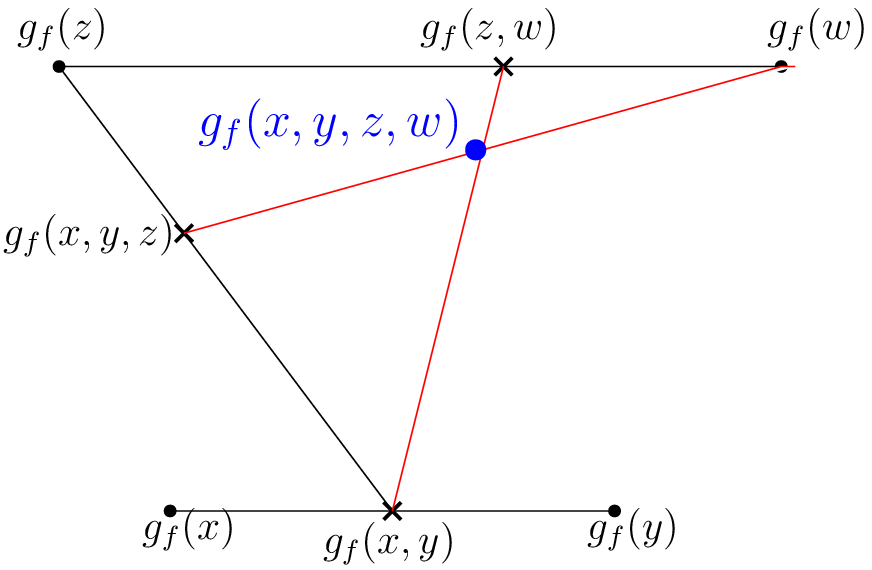}}\par 
\subfloat[$g_f(x,y,w)$.]{\label{fig_example_1_4}\includegraphics[width=.45\linewidth]{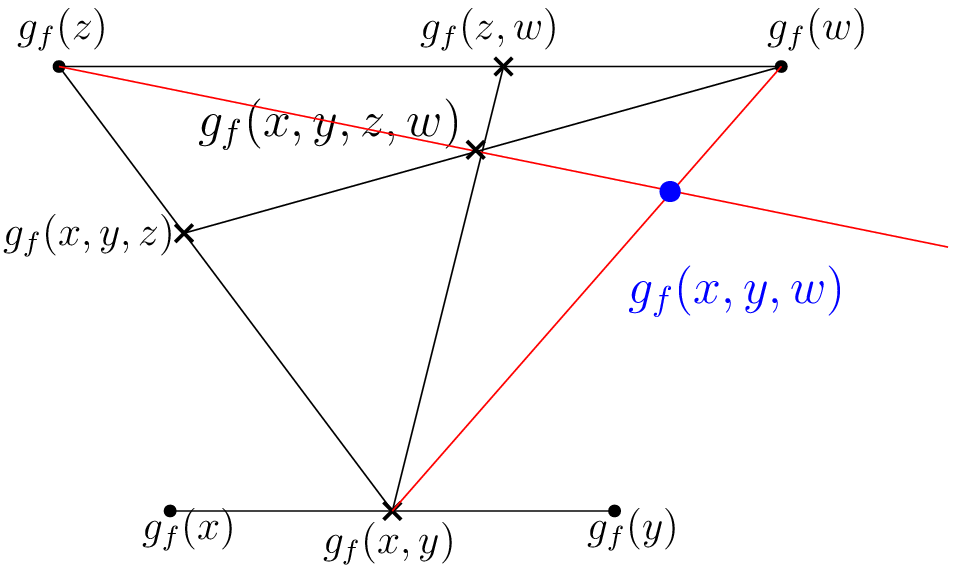}}
\caption{By consistency, we can inductively obtain $g_f(x,y,z), g_f(x,y,z,w)$, and $g_f(x,y,w)$.}
\label{fig_example_1}
\end{figure}

First, as in the figure \ref{fig_example_1_2}, by consistency $g_f(x,y,z)$ is on the intersection of the line joining $g_f(x), g_f(y,z)$ and the line joining $g_f(z),g_f(x,y)$. Then, as in the figure \ref{fig_example_1_3},  again consistency shows that $g_f(x,y,z,w)$ must be the intersection of the line joining $g_f(x,y,z), g_f(w)$ and the line joining $g_f(z,w), g_f(x,y)$. Finally, as in the figure \ref{fig_example_1_4}, one last application of the consistency shows that $g_f(x,y,w)$ must be on the intersection of the line joining $g_f(z), g_f(x,y,z,w)$ and the line joining $g_f(x,y), g_f(w)$. Therefore, by consistency all three $g_f(x,y,z), g_f(x,y,w)$, and $g_f(x,y,z,w)$ are uniquely determined.




\end{example}

By inductively applying the same couple of arguments, as in the Example~\ref{Example_c_vs_cwp} (see \cite{hamzeowhadi1} Thm. 1), we obtain the following more general result.

\begin{corollary}\label{cor_AAR}
Let $g_f:\mathcal{E}^*\to\mathcal{P}(\Theta)$ be a consistent ranking rule. If the range of $g_f$ is not a subset of a one dimensional linear variety, then there exists a weight function $w: \mathcal{E}\to \mathbb{R}_{++}$ such that for every set of experts  $A\in \mathcal{E}^*$, 
\begin{equation}
g_f(A)=\sum\limits_{e_i\in A}\left(\frac{w(e_i)}{\sum\limits_{e_j\in A} w(e_j)}\right)g_f(e_i).
\end{equation}
Moreover, the weight function is unique up to multiplication by a positive number.
\end{corollary}

As a consequence of the above corollary, a modeling rule $f:\mathcal{E}^*\to\Delta(\Theta)$ is consistent (and non-degenerate) if and only if it can be constructed as follows.

\begin{enumerate}

    \item select a weight function $w:\mathcal{E}\to\mathbb{R}_{++}$,
    \item figure out $g_f(e)$ for $e\in \mathcal{E}$,
    \item for every $A\in \mathcal{E}^*$, form $g_f(A)$ as in 
        \begin{equation}
            g_f(A)=\sum\limits_{e_i\in A}\left(\frac{w(e_i)}{\sum\limits_{e_j\in A} w(e_j)}\right)g_f(e_i),
        \end{equation}
    \item finally, the rule $f$ is
        \begin{equation}
            f(A) \in \argmin_{d \in \Delta(D) } R(g_f(A), d).
        \end{equation}

\end{enumerate}

We now present more general version in which, instead of consistency, we impose the weak consistency. Weakly consistent rules are characterized by both a weight function and weak order over experts. They are obtained by averaging the prior of the highest ordered experts rather than all of them.  

\begin{definition}
A binary relation $\succcurlyeq$ on $\mathcal{E}$ is a \textbf{\emph{weak order}} on $\mathcal{E}$, if it is reflexive ($x \succcurlyeq x$), transitive ($x \succcurlyeq y$ and $y \succcurlyeq z$ imply $x \succcurlyeq z$), and complete (for all $x,y \in X$, $x\succcurlyeq y$ or $y \succcurlyeq x$). We say that $x$ is equivalent to $y$, and write $x\sim y$, if $x\succcurlyeq y$ and $y \succcurlyeq x$.

Consider a weak order $\succcurlyeq$ on  $\mathcal{E}$. For $A\in \mathcal{E}^*$, write $M(A,\succsim)$ for the highest order elements in $A$.
\end{definition}

A more general result is as follows (see \cite{hamzeowhadi1} Thm. 2).

\begin{corollary}\label{cor_AAR_general}
Let $g_f:\mathcal{E}^*\to\mathcal{P}(\Theta)$ be a weakly consistent ranking rule. If $g_f$ satisfies the non-degeneracy condition (strongly richness) condition of \cite{hamzeowhadi1}, then there exist a unique weak order $\succcurlyeq$ on $\mathcal{E}$ and a weight function $w: \mathcal{E}\to \mathbb{R}_{++}$ such that for every set of experts $A\in \mathcal{E}^*$, 
\begin{equation}
g_f(A)=\sum\limits_{e_i\in M(A,\succcurlyeq)}\left(\frac{w(e_i)}{\sum\limits_{e_j\in M(A,\succcurlyeq)} w(e_j)}\right)g_f(e_i).
\end{equation}
Moreover, the weight function is unique up to multiplication by a positive number in each of the equivalence classes of the weak order$\succcurlyeq$.
\end{corollary}

As a consequence of the above corollary, a (non-degenerate) modeling rule $f:\mathcal{E}^*\to\Delta(\Theta)$ is weakly consistent if and only if it can be constructed as follows.

\begin{enumerate}

    \item select a weight function $w:\mathcal{E}\to\mathbb{R}_{++}$ and a weak order $\succcurlyeq$ on $\mathcal{E}$,
    \item figure out $g_f(e)$ for $e\in \mathcal{E}$,
    \item for every $A\in \mathcal{E}$, form $g_f(A)$ as in 
        \begin{equation}\label{main_equation2}
            g_f(A)=\sum\limits_{e_i\in M(A,\succcurlyeq)}\left(\frac{w(e_i)}{\sum\limits_{e_j\in M(A,\succcurlyeq)} w(e_j)}\right)g_f(e_i),
        \end{equation}
    \item finally, the rule $f$ is
        \begin{equation}
            f(A) \in \argmin_{d \in \Delta(D) } R(g_f(A), d).
        \end{equation}

\end{enumerate} 

The representation \eqref{main_equation2} has two components: one is captured by the weak order $\succcurlyeq$; the other is the weight function $w$. The weak order partitions the set of experts into equivalence classes and ranks them from top to bottom. If all experts $e\in A$ have the same ranking, then $g_f(A)$ is the weighted average of $g(e)$ for $e\in A$. However, if some experts have a higher ranking than others, then the rule will ignore the lower-ordered experts.
Hence, the assessment of the rule has two steps. First, it only considers the highest-ordered priors. Then, it uses the weight function and finds the weighted average among the highest-ordered priors.\\

\section{More Examples}\label{sec_examles} 

In the previous section, we interpreted the elements of  the set $\mathcal{E}$ as individual experts. We will from now interpret the elements of $\mathcal{E}$ as representing  experts and their characteristics. 

\subsection{Kernel Smoother}
Assume that the decision-maker herself is also an element of the set $\mathcal{E}$. To be precise, let $e_0\in \mathcal{E}$ be all the relevant characteristics and beliefs of the decision-maker without the observation of any other expert. Therefore, with using the same language as before, $g_f(e_0,.):\mathcal{E}^*\cup \emptyset\to\mathcal{P}(\Theta)$ and $g_f(e_0,\emptyset)$ is the prior that the decision-maker is going to use for selection the decision rules, without observing any other experts' characteristics. 

More generally, we can interpret the rule $g_f :\mathcal{E}\times\mathcal{E}^*\cup \emptyset\to\mathcal{P}(\Theta)$ as a ranking rule such
that for every characteristics of the decision maker $e\in \mathcal{E}$ and for every observation of the set of experts' characteristics $A\in \mathcal{E}^*$, $g_f(e,A)$ is the prior that the decision maker is going to use to select her decision rule.

Under the conditions of the previous section, we have the following representation.

\begin{corollary}
Let $g_f :\mathcal{E}\times\mathcal{E}^*\cup \emptyset\to\mathcal{P}(\Theta)$ be a rule such that for every $e\in \mathcal{E}$, $g_f(e,.)$ being a consistent ranking rule with the range not being a subset of a one dimensional linear variety. Then, then  exists a weight function $w: \mathcal{E}\times\mathcal{E}\to \mathbb{R}_{++}$ such that for every decision maker's characteristics $e\in \mathcal{E}$ and for every set of expert's characteristics $A\in \mathcal{E}^*$, 
\begin{equation}
g_f(e,A)=\sum\limits_{e_i\in A}\left(\frac{w(e,e_i)}{\sum\limits_{e_j\in A} w(e,e_j)}\right)g_f(e,e_i).
\end{equation}
Moreover, the weight function is unique up to multiplication by a positive number.

\end{corollary}

As a consequence of the above representation, $w$ behaves as a similarity measure between the decision maker's characteristics and the expert's characteristics. The representation is the same as the \emph{\textbf{Kernel Smoother}} in statistics. More precisely, consider the following problem.

\begin{problem}
Let $D=\{(x_1,y(x_1)) \ldots,(x_N,y(x_N))\}$ be $N$ given sample points with $x_i\in \mathbb{R}^k$ and $y(x_i)\in\mathbb{R}$. Let $x_0\in\mathbb{R}^k$ be another point, the goal is to estimate $y(x_0)$.
\end{problem}

There are many different approaches to the above problem, but one is as follows.

\begin{definition}
 Let $d:\mathbb{R}\to\mathbb{R}$ be a non increasing function, $h:\mathbb{R}^k\to\mathbb{R}$ be a hyper-parameter, and $\|\cdot\|$ be the euclidean norm. Then for every two point $w,z\in \mathbb{R}^k$ a kernel $K_{h,d}$ can be defined as follows

\begin{equation}
    K_{h,d}(w,z)=d\left({\frac {\left\|w-z\right\|}{h(w)}}\right)
\end{equation}
And for every $ x_0\in\mathbb{R}^k$, the \emph{\textbf{Nadaraya-Watson kernel-weighted average}} is defined by
\begin{equation}
y(x_0):=\frac{\sum \limits _{i=1}^{N}{K_{h,d}(x_0,x_{i})y(x_{i})}}{\sum \limits _{i=1}^{N}K_{h,d}(x_{0},x_{i})}
\end{equation}
\end{definition}

To check that the above form is a subset of our representation, we define the set of characteristics $\mathcal{E}=\{U\}\times\mathbb{R}^k \cup \{O\}\times\mathbb{R}^k\times\mathcal{P}(\Theta)$, and interpret the first coordinate of $e\in \mathcal{E}$ as being observed or unobserved, and the other coordinates as the input and the output of the function we are tying to estimates. Then, under the assumption of the consistent rule, there exists a weight function $w: \mathcal{E}\times\mathcal{E}\to \mathbb{R}_{++}$ such that for every decision maker's characteristics $e\in \mathcal{E}$ and for every set of expert's characteristics $A\in \mathcal{E}^*$, 
\begin{equation}
g_f(e,A)=\sum\limits_{e_i\in A}\left(\frac{w(e,e_i)}{\sum\limits_{e_j\in A} w(e,e_j)}\right)g_f(e,e_i).
\end{equation}

We may add the assumption that for every $e\in \mathcal{E}$ such that the first coordinate of $e$ is $O$, we should report the third coordinate as output. That is, for every $e$ such that $e_1 = O$, and for every $e_i\in \mathcal{E}$ with $(e_i)_1 = O$, we set $g(e,e_i)=(e_i)_3$. Then the result is exactly the interpolation of inputs using the similarity kernel defined by the weight function $w$.

Note that by restricting the assumptions on the form of the rule $g_f$, we may obtain a different class of weight functions. For example one can we enforce that for every two characteristics $e_1,e_2\in\mathcal{E}$, $g_f(e_1,\{e_2\})=g_f(e_2,\{e_1\})$, then the weight function $w$ in the representation must be a symmetric one.  That is, if from the decision maker's perspective there is no difference between the decision-maker being $e_1$ and the expert's being $e_2$ or the decision-maker being $e_1$ and the expert being $e_2$, we get a symmetric similarity measure. 
 


\subsection{Different Experts with the Same Set of Characteristics}
In many situations, we want to model different individuals with the same form of characteristics. To to that, we can extend the set of characteristics to be $\mathcal{E}^{\textit{ext}}=\mathbb{N}\times \mathcal{E}$, with $(i,e)\in \mathcal{E}^{\textit{ext}}$ represent an individual with index number $i$ with the pure characteristics $e$. With the same (non-degeneracy and consistency of $g_f$) assumptions as in section \ref{sec_main} and by enforcing that for every $e_1,e_2\in \mathcal{E}^\textit{ext}$ with $(e_1)_2 =(e_1)_2$, $g_f(e_1) = g_f(e_2)$, we get the following representation, that there exists a weight function $w: \mathcal{E}\to \mathbb{R}_{++}$ on the set of pure characteristics such that for every set of expert's extended characteristics $A\in {\mathcal{E}^\textit{ext}}$, 
\begin{equation}
g_f(A)=\sum\limits_{(i,e)\in A}\left(\frac{w(e)}{\sum\limits_{(j,e')\in A} w(e')}\right)g_f((1,e)).
\end{equation}

In other words, we can define the function $N:\mathcal{E}\times\mathcal{E}^{\textit{ext}^*}\to \mathbb{N}\cup \{0\}$ to count the number of appearance of the pure characteristics $e\in \mathcal{E}$ that appear in the set $A\in \mathcal{E}^\textit{ext}$ as $N(e,A)$. Then, the representation may simplify as follows

\begin{equation}
g_f(A)=\sum\limits_{e\in \mathcal{E}}\left(\frac{N(e,A) w(e)}{\sum\limits_{e'\in \mathcal{E}}N(e',A) w(e')}\right)g_f((1,e)).
\end{equation}

\subsection{Timing of Experts}

A different example is when the timing of experts matters. In more precise words, we may assume that the observation of expert's characteristics may have a timestamp, and the ones closer to the time of prediction might be more important. In that case, again we can extend the set of characteristics to have the form of $\mathcal{E}^{\textit{ext}}=\mathbb{R_{+}}\times \mathcal{E}$, where $(t,e)\in \mathcal{E}^{\textit{ext}}$ represents an experts $e$ that presents at the time $t$ before the prediction time. 

We might assume that the shifting of all the expert's timestamps by a constant factor should not affect the final model. If that is the case, we have the following notion of \emph{stationarity} property.

\begin{definition}Let $\operatorname{S}:{\mathcal{E}^{\textit{ext}}}^*\times \mathbb{R_{+}} \to {\mathcal{E}^{\textit{ext}}}^*$ represents a  \textbf{\emph{time shift operator}}. That is $\operatorname{S}(A, c)= \{(t+c, e)| (t,e)\in A\}$, for every   $(A, c)\in {\mathcal{E}^{\textit{ext}}}^*\times \mathbb{R_{+}}$. A \textbf{\emph{stationary}} rule $g_f:\mathcal{E}^{\textit{ext}} \to \mathcal{P}(\Theta)$ is such that
\[g_f(\operatorname{S}(A,c))=g_f(A),\]
for $A\in \mathcal{E}^{\textit{ext}}$ and $c\in \mathbb{N}$.
\end{definition}
As a consequence of the result of corollary \ref{cor_AAR}, we have the following representations (see \cite{hamzeowhadi1} Proposition. 1). 

\begin{proposition}\label{time_stationary}
For a (non-degenerate) consistent stationary ranking rule $g_f:\mathcal{{E}^{\textit{ext}}}^* \to \mathcal{P}(\Theta)$, there exist a unique discount factor $q\in (0,\infty)$ and a unique (up to multiplication by a positive number) weight function $w:\mathcal{E}\to \mathbb{R}_{++}$, such that for all $A\in \mathcal{E}^{\textit{ext}}$

\begin{equation}
g_f(A)=\frac{\sum\limits_{(t,e)\in A}q^{t}w(e)g_f(e)}{\sum\limits_{(t,e)\in A} q^{t}w(e)}.
\end{equation}

\end{proposition}

As a consequence of the representation, under the assumption of the proposition, the weight over a received expert $(t,e)\in A$ can be separated into two separate factors. One is the intrinsic value of her characteristics, captured by $w(e)$. The other one is the role of timing, captured by $q^{t}$. Moreover, the only discounting that captures the role of the timing is the exponential form.
If $q=1$, the timing is not important. Hence, the rule only considers the intrinsic value of each expert. However, when $q\neq 1$, the rule places relatively more ($q\in (0,1)$) or less ($q\in (1,\infty)$) weight on the experts closer to the time of the prediction.

\subsection{Social Choice Functions and Voting Mechanisms}
In this example, we consider the connection of our setup to the social choice literature. In the step three of sec \ref{sec_main}, we interpret the rule $g_f:\mathcal{E}^*\to\mathcal{P}(\Theta)$ as a ranking over the set of risk profiles, $\mathcal{S}$. In the end, the decision-maker should select a rule that minimizes her loss with respect to the induced ranking.

One special case is when we interpret each expert to define its own ranking over the risk set. In other words, we can assume that $\mathcal{E}=\mathcal{P}(\Theta)$. In this case, each expert prefers the decision rules that are better suited for their own ranking. The role of the rule $g_f:\mathcal{E}^*\to\mathcal{P}(\Theta)$ is to attach a social ranking of the set of decision rules. Moreover, the role of the rule $f:\mathcal{E}^*\to\Delta(\Theta)$ is to select the socially acceptable rule. 

In this setup, the role of the rule $g_f$ is a voting mechanism that, based on the expert's reported preferences, should attach a social ranking over the set of decision rules. In the end, the final rule is the highest-ranked decision rule based on social rank.

In a more adapted setup, we may assume that $g_f(e) = e$ for all $e \in \mathcal{E}$. That is, if there is only one individual in the population, the final ranking should be that person's ranking. The form of the voting mechanism that works in this setup is as follows.

\begin{corollary}
Let $g_f:\mathcal{P}(\Theta)^*\to\mathcal{P}(\Theta)$ be individualistic consistent ranking rule. Then, there exists a weight function $w:\mathcal{E}\to\mathbb{R}_{++}$

\begin{equation}
g_f(A)=\frac{\sum\limits_{e\in A}w(e)e}{\sum\limits_{e\in A} w(e)}.
\end{equation}

\end{corollary}

As a result, the mechanism is as follows.

\begin{enumerate}
    \item select the weight of each individual,
    \item ask individuals to report their ranking,
    \item form the social ranking by randomized ranking or weighted average of individual's ranking (depends on the application),
    \item reports the best alternative for the social ranking.
\end{enumerate}
Adding other assumptions, like anonymity and different weak versions of independence axiom can explain the relative utilitarianism and other social choice functional forms.

\section{Related Literature}
Our paper is an application of \cite{hamzeowhadi1} in the setting of Wald's statistical decision theory  \cite{Wald:1950}.

The first two steps of our paper (section \ref{step1} and \ref{step2}) rely on the notion of admissibility and the complete class theorem. A simple form of the complete class theorem and admissible rules appears in \cite{lehman:1947} and \cite{Wald:1950}. \cite{Lacam:1955}, \cite{Brown:1971, Brown:1986} provide more general versions of the result. A good reference for our environment of decision making is \cite{lehman_casella} and \cite{berger2013statistical}. 

The complete class theorem connects the notion of admissible  models of frequentist statistics to priors in Bayesian statistics. The result of \cite{owhadi2015brittleness} shows that the posterior can be highly sensitive to the choice of the prior. There are many interesting approaches for selecting a good prior for computation due to the complete class theorem.   One approach is to select a minimax rule, which is the best response to the prior with the most risk for the decision-maker (least favorable prior). \cite{stark1,stark2} provide a computational approach for selecting a near-optimal prior. \cite{OwhadiScovelMachineWald} has a roadmap on how to combine the complexity of the computation, robustness, and the accuracy of the prediction as a general way of using the minimax criteria. A recent approach is to select a prior from a smaller set of feasible priors that make sense for the modeler as in \cite{uq4}. \cite{stark:2015} is a good reference for the connections and shortcomings of the frequentist and Bayesian statistics.

Regarding the third step (section \ref{step3}), \cite{VonNeumann:1944} provides the behavioral justification of ranking the risk set linearly. There are many papers on more general versions and shortcomings of linear preferences, from Choquet expected utility theory of \cite{schmeidler:1989} to prospect theory of \cite{kahneman_tversky:1979}. \cite{gilboa_monographs} is a great reference for this regard. 

In the fourth step (section \ref{step4}), the notion of consistency is a particular case of the one that appears in \cite{hamzeowhadi1}. Another general approach is through the theory of Cased-Based Prediction developed by the seminal works of \cite{gilboa1,gilboa_inductive, gilboa_book} and \cite{billot1}. 

There are many similar notions of consistency in different environments of economics and statistics. In the context of social choice, \cite{shapley2}, \cite{dhillon}, and \cite{shapley1} study variants of extended Pareto rules. 
In the context of choice theory, \cite{ahn} studied a similar form of rules. The path independence choice functions are extensively studied by \cite{plott}. \cite{hamzeowhadi1} has a complete list of the relevant topics. 

Kernel Smoothing and generally smoothing methods have a long history in statistics. A good introduction of the topic is chapter four of \cite{wasserman:nonparametric}. For a more in-depth study see \cite{scott:1992} and \cite{ruppert:2003}. For the axiomatic approaches, \cite{gilboa_inductive} is a great reference. 

Finally, our goal is only to provide a logical foundation for the hierarchical Bayes method. Many papers are discussing the method and applications in different fields of sciences. By searching the Hierarchical Bayes method in google scholar, one can find more than 100,000 scholarly articles on the topic!

\subsection*{Acknowledgments}
The authors gratefully acknowledge support
from Beyond Limits (Learning Optimal Models) through CAST (The Caltech Center for Autonomous Systems and Technologies) and partial support
from the Air Force Office of Scientific Research under awards number FA9550-18-1-0271 (Games for Computation and Learning) and FA9550-20-1-0358 (Machine Learning and Physics-Based Modeling and Simulation).


\section*{References}

\printbibliography[heading=none]

@Book{Aliprantis2006,
  Title                    = {Infinite {D}imensional {A}nalysis: {A} {H}itchhiker's {G}uide},
  Author                   = {Aliprantis, C. D. and Border, K. C.},
  Publisher                = {Springer},
  Year                     = {2006},
  Address                  = {Berlin},
  Edition                  = {Third},
}

@Book{VonNeumann:1944,
  Title                    = {Theory of {G}ames and {E}conomic {B}ehavior},
  Author                   = {von-Neumann, J. and Morgenstern, O.},
  Publisher                = {Princeton University Press, Princeton, New Jersey},
  Year                     = {1944},

  Mrclass                  = {90.0X},
  Mrnumber                 = {0011937 (6,235k)},
  Mrreviewer               = {A. Wald},
  Pages                    = {xviii+625}
}

@InProceedings{OwhadiScovelMachineWald,
  Title                    = {Toward {M}achine {W}ald},
  Editor                   = {Ghanem R., Higdon D., Owhadi H.},
  Author                   = {Owhadi, H. and Scovel, C.},
  Booktitle                = {Handbook of Uncertainty Quantification},
  Publisher                ={Springer},
  Year                     = {2017},
  Pages                    = {157--191},
  Note                     = {arXiv:1508.02449}
}

@Book{Wald:1950,
  Title                    = {Statistical {D}ecision {F}unctions},
  Author                   = {Wald, A.},
  Publisher                = {John Wiley \& Sons Inc.},
  Year                     = {1950},

  Address                  = {New York, NY},

  Mrclass                  = {62.0X},
  Mrnumber                 = {0036976 (12,193f)},
  Mrreviewer               = {D. Blackwell},
  Pages                    = {ix+179}
}

@Article{Wald:essentially,
  Title                    = {An essentially complete class of admissible decision functions},
  Author                   = {Wald, A.},
  Journal                  = {Ann. Math. Statist.},
  Year                     = {1947},
  Pages                    = {549--555},
  Publisher                = {JSTOR}
}

@book{berger2013statistical,
  title={Statistical Decision Theory and Bayesian Analysis},
  author={Berger, J. O.},
  year={2013},
  publisher={Springer Science \& Business Media}
}

@article{owhadi2015brittleness,
  title={On the brittleness of {B}ayesian inference},
  author={Owhadi, H. and Scovel, C. and Sullivan, T.},
  journal={SIAM Review},
  volume={57},
  number={4},
  pages={566--582},
  year={2015},
  publisher={SIAM}
}

@book{luenberger1969optimization,
  title={Optimization by Vector Space Methods},
  author={Luenberger, D. G.},
  year={1969},
  publisher={John Wiley \& Sons}
}

@article{ahn,  
  author = "Ahn, David. S. and Echenique, F. and Saito, K.",
   title = "On Path Independent Stochastic Choice",
 journal = "Theoretical Economics",
  volume = 13,
    year = 2018,
   pages = "61--85"}

@article{shapley1,
  author = "Baucells, M. and Shapley, L. S.",
   title = "Multiperson utility",
 journal = "Games and Economic Behavior",
  volume = 62,
    year = 2008,
   pages = "329--347"}

@article{billot1,
  author = "Billot, A. and Gilboa, I. and Samet, D. and Schmeidler, D.",
   title = "Probabilities as Similarity-Weighted Frequencies",
 journal = "Econometrica",
  volume = 73,
    year = 2005,
   pages = "1125--1136"}

@article{dhillon,
  author = "Dhillon, A.",
   title = "Extended Pareto Rules and Relative Utilitarianism",
 journal = "Social Choice Welfare",
  volume = 15,
    year = 1998,
   pages = "521--542"}

@article{gilboa1,
  author = "Gilboa, I. and Schmeidler, D.",
   title = "Case-Based Decision Theory",
 journal = "Quarterly Journal of Economics",
  volume = 110,
    year = 1995,
   pages = "605--639"}

@article{plott,
  author = "Plott, C. R.",
   title = "Path Independence, Rationality, and Social Choice",
 journal = "Econometrica",
  volume = 41,
    year = 1973,
   pages = "1075--1091"}

@book{gilboa_book,
  author    = {Gilboa, I. and Schmeidler, D.}, 
  title     = {Case-Based Predictions: An Axiomatic Approach to Prediction, Classification and Statistical Learning},
  publisher = {World Scientific Publishing Co, Singapore.},
  year      = 2012,
  edition   = 1
}

@article{gilboa_inductive,
  author = "Gilboa, I. and Schmeidler, D.",
   title = "Inductive Inference: An Axiomatic Approach",
 journal = "Econometrica",
  volume = 71,
    year = 2003,
   pages = "1--26"}

@incollection{shapley2,
  author       = {Shapley, L. S. and Shubik, M.}, 
  title        = {Preferences and Utility},
  booktitle    = {Game Theory in the Social Sciences: Concepts and Solutions},
  publisher    = {MIT Press},
  year         = 1982,
  editor       = {Shubik, M.},
  chapter      = 4,
  address      = {Cambridge, MA},
  edition      = 3,
  month        = 7,

}

@Article{hamzeowhadi1,
  Title                    = {Aggregation of models, choices, beliefs, and preferences},
  Author                   = {Hamze Bajgiran, H. and Owhadi, H.},
  Journal                  = {arXiv},
  Year                     = {2021},
  Note                     = {\url{https://arxiv.org/abs/2111.11630}},
}

@book{gilboa_monographs,
  author    = {Gilboa, I.}, 
  title     = {Theory of Decision under Uncertainty},
  publisher = {Cambridge University Press},
  year      = 2009,
  edition   = 1
}

@Article{kahneman_tversky:1979,
  Title                    = {Prospect Theory: An Analysis of Decision under Risk},
  Author                   = {Kahneman, D. and Tversky, A.},
  Journal                  = {Econometrica},
  volume                   = {73},
  Year                     = {1979},
  Pages                    = {363--291},
}

@book{lehman_casella,
  author    = {Lehmann, E. L. and Casella, G.}, 
  title     = {Theory of Point Estimation},
  publisher = {Springer},
  year      = 1998,
  edition   = 2
}

@Article{lehman:1947,
  Title                    = {On Families of Admissible Tests},
  Author                   = {Lehmann, E. L.},
  Journal                  = {Ann. Math. Statist},
  volume                   = {18},
  Year                     = {1947},
  Pages                    = {97–104},
}

@Article{Lacam:1955,
  Title                    = {An Extension of Wald's Theory of Statistical Decision Functions},
  Author                   = {Le Cam, L.},
  Journal                  = {Ann. Math. Statist},
  volume                   = {26},
  Year                     = {1955},
  Pages                    = {69–81},
}

@Article{Brown:1971,
  Title                    = {Admissible estimators, Recurrent diffusions, and insoluble boundary value problems},
  Author                   = {Brown, L. D.},
  Journal                  = {Ann. Math. Statist},
  volume                   = {42},
  Year                     = {1971},
  Pages                    = {855–904},
}

@Article{Brown:1986,
  Title                    = {Fundamentals of Statistical Exponential Families with Applications in Statistical Decision Theory},
  Author                   = {Brown, L. D.},
  Journal                  = {Lecture Notes-Monograph Series, Published by: Institute of Mathematical Statistics},
  volume                   = {9},
  Year                     = {1986},
  Pages                    = {i-iii+vvii+
ix-x+1-279},
}

@Article{stark:2015,
  Title                    = {Constraints versus Priors},
  Author                   = {Stark, P. B.},
  Journal                  = {SIAM/ASA J. UNCERTAINTY QUANTIFICATION},
  volume                   = {3},
  Year                     = {2015},
  Pages                    = {586-598},
}

@Article{uq4,
  Title                    = {Uncertainty Quantification of the 4th kind; optimal posterior accuracy-uncertainty tradeoff with the minimum enclosing ball},
  Author                   = {Hamze Bajgiran, H. and Batlle Franch, P. and Owhadi, H and Scovel, C and Shirdel, M. and Stanley, M and Tavallali, P.},
  Journal                  = {arXiv},
  Year                     = {2021},
  Note                     = {\url{https://arxiv.org/abs/2108.10517}},
}

@Article{stark1,
  Title                    = {Using what we know: Inference with physical constraints},
  Author                   = {Schafer, C.M. and  Stark, P.B.},
  Journal                  = {in Proceedings of the Conference on Statistical Problems in Particle Physics, Astrophysics and Cosmology,
PHYSTAT2003, Menlo Park, CA, L. Lyons, R. Mount, and R. Reitmeyer, eds.},
  Year                     = {2003},
  Pages                    = {25–34},
}

@Article{stark2,
  Title                    = {Constructing confidence sets of optimal expected size},
  Author                   = {Schafer, C.M. and  Stark, P.B.},
  Journal                  = {J. Amer. Statist.
Assoc.},
  volume                   = {104},
  Year                     = {2009},
  Pages                    = {1080-1089},
}

@Article{schmeidler:1989,
  Title                    = {Subjective Probability and Expected Utility without Additivity},
  Author                   = {Schmeidler, D.},
  Journal                  = {Proceedings of the American Mathematical Society},
  volume                   = {197},
  Year                     = {1989},
  Pages                    = {255-261},
}

@Book{scott:1992,
  Title                    = {Multivariate Density Estimation: Theory, Practice, and Visualization},
  Author                   = {Scott, D.W.},
  Publisher                = {Wiley. New York, NY.},
  Year                     = {1992},
}

@Book{ruppert:2003,
  Title                    = {Semiparametric Regression},
  Author                   = {Ruppert, D. and Wand, M. P.},
  Publisher                = {Cambridge University Press. Cambridge.},
  Year                     = {2003},

}

@Book{wasserman:nonparametric,
  Title                    = {All of Nonparametric Statistics},
  Author                   = {Wasserman, L.},
  Publisher                = {Springer},
  Year                     = {2005},

}



\end{document}